\documentclass[a4paper, 12pt]{article}

\usepackage[dvips]{graphics,color}
\def\spc{\color[rgb]{1,0.2,0.2}}
\def\nc{\normalcolor}
\usepackage[left=1in,right=1in,top=1in,bottom=1in]{geometry}

\usepackage[english]{babel}
\usepackage{amsmath, amsthm, amsfonts, amsbsy, amssymb, scrextend, graphicx}
\usepackage{hyperref} 
\usepackage{verbatim}
\usepackage{color}

\usepackage{tikz}
\usepackage[export]{adjustbox}
\usepackage{float}

\newtheorem {corollary}{Corollary}
\newtheorem {proposition}{Proposition}
\newtheorem {theorem}{Theorem}

\newtheorem{innercustomgeneric}{\customgenericname}
\providecommand{\customgenericname}{}
\newcommand{\newcustomtheorem}[2]{%
	\newenvironment{#1}[1]
	{%
		\renewcommand\customgenericname{#2}%
		\renewcommand\theinnercustomgeneric{##1}%
		\innercustomgeneric
	}
	{\endinnercustomgeneric}
}

\newcustomtheorem{customthm}{Theorem}

\newtheorem{remark}{Remark}

\newtheorem{example}{Example}

\newcommand{\R}{{\mathbb{R}}}

\newcommand{\Nc}{\mathcal{N}}
\renewcommand{\P}{\mathsf{P}}

\newcommand{\C}{\mathsf{C}_{\rm lvm}}
\newcommand{\Put}{\mathsf{P}_{\rm lvm}}

\newcommand{\CBS}{\mathsf{C}_{\rm BS}}
\newcommand{\PBS}{\mathsf{P}_{\rm BS}}

\newcommand{\Vatm}{\mathsf{V}_{\rm atm}}
\newcommand{\BS}{\mathsf{BS}_{\rm atm}}
\newcommand{\sigatm}{\sigma_{\rm atm}}

\newcommand{\eps}{\varepsilon}
\newcommand{\E}{\mathsf{E}}

\newcommand{\sign}{{\rm sgn}}

\newcommand{\wt}{\widetilde}

\newcommand{\sgm}{\sigma_{\scriptscriptstyle -}}
\newcommand{\sgp}{\sigma_{\scriptscriptstyle +}}

\newcommand{\sigmaBS}{\sigma_{ {\rm BS}}}

\newcommand{\mathsym}[1]{{}}

\newcommand{\unicode}[1]{{}}

\begin{document}

\title{
Extreme ATM skew 
 in a local volatility model with discontinuity: joint density approach
}

\author{
Alexander Gairat\footnote{Gearquant.
 Email address: agairat@gearquant.com
}\, and 
Vadim Shcherbakov\footnote{
 Royal Holloway,  University of London.
 Email address: vadim.shcherbakov@rhul.ac.uk
}
}

\maketitle


\begin{abstract}
{\small 
This paper concerns  a local volatility model in which 
 volatility takes two possible values, and the specific value depends on whether the underlying  price is above or below a given  threshold value. The 
 model is known, and a number of results have been obtained for it. 
In particular, option pricing formulas and 
a power law behaviour of the  implied volatility skew have been established
in the case when  the threshold is taken at the money. 
In this paper we  derive an alternative representation of
 option pricing formulas.
 In addition, we obtain an approximation of 
option prices by the corresponding Black-Scholes prices. Using this approximation 
streamlines obtaining the aforementioned behaviour of the skew.
Our approach is  based on the 
natural relationship of the  model with Skew Brownian motion and 
consists  of the systematic use  of the joint distribution of this stochastic process and some of its functionals.
}
\end{abstract}

\noindent {{\bf Keywords:} 
 Local volatility model,  Skew Brownian motion, implied volatility,\\
  at the money skew}

\section{Introduction}

This  paper concerns  a  local volatility model (LVM),  
in which volatility  takes only  two possible  values. 
If the underlying price is larger or equal to some threshold value $R$, then volatility is equal to $\sgp$, and
if the underlying price is less than $R$, then volatility is equal to $\sgm$, where $\sgp$ and $\sgm$ are given
positive constants.
In what follows we refer  to this model as the two-valued LVM.
If $\sgp=\sgm=\sigma$, then the model is just the classic log-normal model with constant volatility $\sigma$,
under which the celebrated Black-Scholes (BS)  formula for the price of a European option has been obtained.

The two-valued LVM is well known, and a number of results for the model are available
 (e.g. see Gairat and Shcherbakov~\cite{GS}, Lipton~\cite{Lipton1}, Lipton and Sepp~\cite{Lipton2}, Pigato~\cite{Pigato} and references therein).  
 Pricing formulas for European  options
 have been obtained in Gairat and Shcherbakov~\cite{GS}
 for arbitrary spot  $S_0$
 and strike $K$.
In Pigato~\cite{Pigato}  option pricing formulas have been 
obtained in the case  when $S_0=R$ or $K=R$.  The other results in
that paper concern the analysis of the implied volatility surface  in the case when  $S_0=R$.
 In particular, it was shown  that if   strike $K$ and
 maturity  $T$ are related 
 by the equation $K=e^{\gamma\sqrt{T}}$ (``the central limit regime''), 
 then  implied volatility $\sigmaBS(T,K)=\sigmaBS(T,e^{\gamma\sqrt{T}})$ 
 converges to a smooth function $\sigmaBS(\gamma)$ of $\gamma$, as $T\to 0$, 
 and Taylor's expansion of the second order for this limit function was explicitly computed.
Moreover, an exact formula for the  at the money (ATM) implied volatility 
skew was obtained.  
These results were  then used to show  that the ATM skew 
 behaves  like $T^{-1/2}$, as  $T\to 0$. 

The result concerning the short term behaviour of the ATM skew 
 is of particular interest for the following reasons.
On one hand, this reproduces the well-known power law behaviour of 
the skew for short term maturities observed in some data   (see Pigato~\cite{Pigato} and references therein for more details). On the other hand, 
 in Guyon and El Amrani~\cite{Guyon} the authors  claim that
``our study suggests that the term-structure of equity ATM
skew has a power-law shape for maturities above $1$ month but has a different behavior, and in particular may not blow up, for shorter maturities''.
This difference in results and opinions should be taken into account when using the two-valued LVM.  In addition, note that 
in Foreign Exchange option markets the standard practice is  
to quote  implied volatility in terms of delta moneyness  
\(\frac{\log(K/F)}{\sigma \sqrt{T}}\)  
 (e.g. see Clark~\cite[Chapter 3]{Clark} and references therein).
  As a result, the  skew  asymptotically behaves as 
 \(T^{-1/2}\) when the implied volatility  is  quoted in terms of the 
 moneyness  \(\log(K/F)\).

The research approach in Pigato~\cite{Pigato} was based on 
using the  Laplace transform and related research techniques.
In the present paper we  demonstrate another approach 
to the study of the two-valued LVM.
 This alternative method   is  based on  the natural relationship 
of the two-valued LVM with Skew Brownian motion (SBM). The latter is 
 a continuous time Markov process 
obtained from standard Brownian motion (BM) by independently choosing with certain fixed probabilities the signs 
of the excursions away from the origin. 
If these probabilities are equal to $1/2$, then the process is standard BM.
It turns out that if the underlying price follows the two-valued LVM, then the natural logarithm of the price divided by volatility is a special case of SBM with a two-valued 
drift (to be explained). 
Our approach  consists of using  the joint distribution of this process and its 
functionals,  such  as the local time at the origin, the last visit to the origin and the  occupation time. The  distribution  was obtained in
 Gairat and Shcherbakov~\cite{GS}, where it was
applied   to option valuation under the two-valued LVM.
 Using this distribution allowed  to streamline some computations 
 in  a special case of  SBM  in Gairat and Shcherbakov~\cite{DryFriction}.
 In the present paper, we give another example of the 
 application of the joint  distribution.

 First of all,  we use the distribution
  to obtain  option pricing formulas   in the   case $S_0=R=1$
  (the assumption $R=1$ is  just a technical one, as the  general case of $R$ 
can  be readily reduced  to this one by dividing the underlying price by $R$).  
As we mentioned above,  the option pricing in the general case of the 
initial underlying price was considered in~\cite{GS}  
(see Section~\ref{Prices} below  for details).
However, the case $S_0=R=1$  
 was not explicitly mentioned in that paper.
Although  pricing formulas in this case
can be obtained   by passing to the limit $S_0\to 1$ 
in more general formulas in~\cite{GS},
we derive them  here directly by using the aforementioned joint distribution
(as we did in~\cite{GS} in the general case). This allows us to once
 again demonstrate the proposed method.
Besides, the corresponding computations are simplified in the  case when the discontinuity threshold is taken at the money.
We  also obtain a new representation of 
the option prices in this case by combining the joint  distribution with the well-known Dupire's forward equation. 
These new  option pricing formulas are in the form of the convolution 
of ATM prices with the density of the first passage time to zero of the standard Brownian motion, which is easy to interpret in probabilistic terms. 
Furthermore, we use our distributional results 
for  obtaining the approximation of option prices in the two-valued LVM by the corresponding BS prices (BS approximation). 
Briefly, the BS-approximation is as follows.  Consider a European option with maturity $T$ and strike $K\geq 1$.
 Let $\C(K, T)$ be the option price under the two-valued LVM and let  $\CBS(\sgp,K, T)$ be the BS price of 
 the same option, when volatility is equal to $\sgp$.
 Then $|\C(K,T)-2p\CBS(\sgp, K, T)|\leq cT$ for all sufficiently small $T$, where $p=\frac{\sgm}{\sgm+\sgp}$.
 A similar approximation holds for prices of European put options.
 The BS approximation was already briefly  noted in~\cite{GS},
  although it was not sufficiently  appreciated there. 
In the present paper we discuss the approximation in more detail
and apply it to obtain some key results    concerning the short term 
behaviour of  implied  volatility. 
In addition, we obtain a new form of the approximation 
and apply it to estimate  implied volatility. 
  
The rest of the paper is organised as follows. In Section~\ref{LVM} we formally define the two-valued LVM and state the key result concerning  the aforementioned joint distribution.
Option pricing formulas and  the BS approximation 
are stated in Section~\ref{Prices}.
Section~\ref{Volatility} concerns the analysis of the implied volatility surface. 
Proofs of most of the  results  are given in Section~\ref{A}.
In Section~\ref{BS-study} a modification of the  BS approximation is discussed.

\section{The model }
\label{LVM}

Start with some notations.
Let $(\Omega, {\cal F}, \P)$
  be a probability space on which all random variables under
consideration are defined. 
The expectation with respect to the probability measure $\P$ will be 
 denoted by $\E$. 
Throughout  $W_t=(W_t,\, t\geq 0)$ denotes  standard Brownian motion (BM), 
and ${\bf 1}_{A}$ denotes  the indicator function of a set or an event $A$.

Without loss of generality we assume throughout that the threshold value $R$ of the underlying price, where the volatility changes its value, is $R=1$.  
 
In the two-valued LVM that was briefly described in the introduction 
 the underlying price 
  $S_t=(S_t,\, t\geq 0)$  follows the equation
\begin{equation*}
dS_t=\sigma(\log S_t) S_t dW_t,
\end{equation*}
where  the  function $\sigma$ is given by 
\begin{equation*}
\sigma(x)=\sgp{\bf 1}_{\{x\geq 0\}}+\sgm{\bf 1}_{\{x<0\}}
\end{equation*} 
for some constants $\sgp>0$ and $\sgm>0$.

Further, consider the process  $X_t=(X_t,\, t\geq 0)$ defined by 
\begin{equation}
\label{X_t}
X_t=\begin{cases}
\frac{\log S_t}{\sgp}&\text{ for } S_t\geq 1,\\
\frac{\log S_t}{\sgm}&\text{ for }  S_t<1.
\end{cases}
\end{equation}
By~\cite[Lemma 1]{GS}), the process $X_t$ follows the equation
\begin{equation}
\label{X_SDE}
dX_t=m(X_t)dt+(p-q)dL_t+dW_t,
\end{equation}
where 
\begin{align}
\label{m(x)}
m(x)&=-\frac{\sgp}{2}{\bf 1}_{\{x\geq 0\}}-\frac{\sgm}{2}{\bf 1}_{\{x<0\}},\\
\label{pq}
p&=\frac{\sgm}{\sgm+\sgp},\\
\label{qq}
q&=1-p=\frac{\sgp}{\sgm+\sgp},
\end{align}
and 
\begin{equation}
\label{L}
L_T=L^{(X)}_T=\lim\limits_{\eps\to 0}\frac{1}{2\eps}\int\limits_0^T {\bf 1}
_{\{-\eps\leq X_u\leq \eps\}}du
\end{equation}
 is  the local time of the process  at the origin.
 
 The process $X_t$ is a special case of Skew Brownian motion (SBM) with a two-valued drift. Recall that SBM
 (without drift)  is 
 obtained from standard  BM  by independently choosing with certain probabilities the signs 
of the excursions away from the origin. 
 An excursion is chosen to be positive with a  fixed  probability $p$
and negative with probability $q=1-p$ (if these probabilities are equal to $\frac{1}{2}$, then the process is standard BM).
The process $X_t$ is SBM with probabilities $p$ and $q$  given by~\eqref{pq} and~\eqref{qq} respectively,
and the two-valued drift~\eqref{m(x)} (see Appendix~\ref{B}).

As we already mentioned in the introduction, a key ingredient 
in our analysis of the two-valued LVM  
is the  use of the joint distribution of the process $X_t$ and some of its functionals.
These functionals include the local time of the process and 
the following quantities.
Namely, given $T>0$ let 
\begin{align}
\nonumber
\tau_0&= \min\left\{t\in [0,T]: X_t=0\right\},\\
\label{tau}
\tau&=\max\{t\in (0,T]: X_t=0\},
\end{align}
be  the  first and the last visits to the origin respectively (on the interval $[0,T]$),
and let 
\begin{align}
\label{V}
V&=\int\limits_{\tau_0}^{\tau}{\bf 1}_{\{X_t\geq 0\}}dt
\end{align}
be the occupation time of the non-negative half line on the interval $[\tau_0, \tau]$.

The distribution of interest is given  in 
Theorem~\ref{C1} below. 
Note that the theorem 
is a special case of a more general result for SBM 
 obtained in~\cite{GS} (see Theorem~\ref{T1} in Appendix~\ref{B}).

\begin{theorem}
\label{C1}
Let $X_t$ be the process defined in~\eqref{X_t}.
Let $L_T$, $\tau$ and  $V$
 be quantities of this process  defined by equations~\eqref{L},~\eqref{tau}
and~\eqref{V} respectively, and let $X_0=0$.
Then, the joint density of random variables
$\tau, V, X_T$ and $ L_T$ is  given by 
\begin{equation}
\label{f-1}
f_{T}(t, v, x, \ell)
=2\alpha(x) h(v, \ell p)h(t-v, \ell q)h(T-t, x)
e^{-\frac{\sgp^2v+\sgm^2(t-v)+\sigma^2(x)(T-t)}{8}-\frac{\sigma(x)}{2}x},
\end{equation}
for $ 0\leq v\leq  t\leq T$ and $\ell\geq 0$, where 
\begin{equation}
\label{a(x)}
\alpha(x)=p{\bf 1}_{\{x\geq 0\}}+q{\bf 1}_{\{x<0\}},
\end{equation}
probabilities $p$ and $q$ are defined in~\eqref{pq} and~\eqref{qq} respectively, 
and 
\begin{equation*}
h(s, y)=\frac{|y|}{\sqrt{2\pi s^3}}e^{-\frac{y^2}{2s}},\,\, y\in \R,  s\in \R_{+},
\end{equation*}
is the  density  of the first passage time  to zero of the standard BM starting at $y$. 
\end{theorem}

\begin{remark}
{\rm
Note that it is convenient to rewrite the density in terms of the  variables 
$u=t-v$ and $s=T-t$ for $t\in[0,T]$.
It is easy to see that if $X_0=0$ and $x\geq 0$, then 
the variable $u$ is the occupation time of the negative 
half-line, and the variable $v+s$ is the total occupation time of the positive half-line, and 
 if $X_0=0$ and $x<0$, then the occupation time of the negative 
half-line is $u+s$, and  the total occupation time of the positive half-line is $v$.
In these terms we have that 
\begin{equation}
\label{f-v-u-s-1}
\begin{split}
f_{T}(t, v, x, \ell)&=\wt f_T(v, u, s, x, \ell)\\
&=2\alpha(x) h(v, \ell p)h(u, \ell q)h(s, x)
e^{-\frac{\sgp^2v+\sgm^2u+\sigma^2(x)s}{2}-\frac{\sigma(x)}{2}x},
\end{split}
\end{equation}
for $(v,u,s): v, u, s\geq 0,\, v+u+s=T$ and $\ell\geq 0$.
}
\end{remark}

\begin{theorem}
\label{L0}
Let $X_t$ be the process defined in~\eqref{X_t}.
Let $p(0,x,T)$ be the probability density function  of  $X_T$  given that  $X_0=0$.
Then 
\begin{align*}
p(0,x,T)&=
2\alpha(x)e^{-\frac{\sigma(x)}{2}x}\int_0^T\phi(T-s)h(s,x)e^{-\frac{\sigma^2(x)}{8}s}ds\\
&=\begin{cases}
2pe^{-\frac{\sgp}{2}x}\int_0^T\phi(T-s)h(s,x)e^{-\frac{\sgp^2}{8}s}ds& \text{ for } x\geq 0,\\
2qe^{-\frac{\sgm}{2}x}\int_0^T\phi(T-s)h(s,x)e^{-\frac{\sgm^2}{8}s}ds& \text{ for } x<0,\\
\end{cases}
\end{align*}
where 
\begin{equation}
\label{phi}
\begin{split}
\phi(t)&=
\frac{1}{\sqrt{2\pi t}(\sgp-\sgm)}
\Big(\sgp e^{-\frac{\sgm^2}{8}t}-\sgm e^{-\frac{\sgp^2}{8}t}\Big)\\
&+
\frac{1}{2}\frac{\sgp\sgm}{(\sgp-\sgm)}
\bigg(\Nc\Big(\frac{\sqrt{t} \sgm}{2}\Big)-\Nc\Big(\frac{\sqrt{t}\sgp}{2}\Big)\bigg)
\end{split}
\end{equation}
and 
\begin{equation*}
\Nc(z)=\int_{-\infty}^z\frac{1}{\sqrt{2\pi}}e^{-\frac{x^2}{2}}dx 
\end{equation*}
for  $z\in \R$.
\end{theorem}	

Theorem~\ref{L0} is proved in Section~\ref{A}.

\begin{remark}
{\rm 
The function $\phi$ defined in~\eqref{phi} has a natural  probabilistic 
sense.
Recall that the function $h(s,x)$ is the density of the 
first passage time to $0$ of the standard BM starting at $x$. The distribution of the first passage  
time  converges, as $x\to 0$,  to the distribution concentrated at $0$. 
It follows from this fact and Theorem~\ref{L0} that 
$$\phi(T)=\frac{1}{2p}\lim\limits_{x\downarrow 0}p(0,x,T)=
\frac{1}{2q}\lim\limits_{x\uparrow 0}p(0,x,T).$$
}
\end{remark}

\section{Option valuation}
\label{Prices}

Let us briefly recall results of~\cite{GS} concerning the option valuation under the two-valued LVM.
Pricing  formulas  were obtained in that paper for knock-in call options  in the  cases
$R=1, S_0\geq 1, K>1$ and $R=1, S_0<1, K>1$.
By combining these results   with the Black-Scholes prices  for  knock-out call options
one can obtain prices of European call options for  other values of $S_0$ and $K$.
The pricing formula in the case $R=1, S_0\geq 1, K>1$  in~\cite{GS} is given in terms of a single integral, 
where an integrand  is analytically  expressed  in terms of the cumulative distribution function (cdf) 
of the standard normal distribution $N(0,1$).
In the special case $R=S_0=1, K>1$ one gets the price of the European option call option.
This case is considered  in Theorem~\ref{Call-Put} below.
The pricing formula in the case $R=1, S_0<1, K>1$ in~\cite{GS}  is also given in terms of a single integral
with an integrand analytically  expressed  in terms of the standard normal distribution   and 
a bivariate normal distribution.

\subsection{Pricing formulas} 


Let 
\begin{equation}
\label{psi}
\psi(a, s, k)=\begin{cases}
\int\limits_{k}^{\infty}e^{ax}h(s,x)dx
=\frac{1}{\sqrt{2\pi s}}e^{ak-\frac{k^2}{2s}}+a\,e^{\frac{a^2}{2}s}
\Nc\Big(\frac{as-k}{\sqrt{s}}\Big)& \text{ for }  k\geq 0,\\
\int\limits_{-\infty}^ke^{ax}h(s,x)dx
=\frac{1}{\sqrt{2\pi s}}e^{ak-\frac{k^2}{2s}}-a\,e^{\frac{a^2}{2}s}
\Nc\Big(\frac{k-as}{\sqrt{s}}\Big)& \text{ for }  k<0.
\end{cases}
\end{equation}
Note that the equation for the function $\psi$ can be rewritten as follows
\begin{equation*}
\psi(a,s,k)=\frac{1}{\sqrt{2\pi s}}e^{ak-\frac{k^2}{2s}}
+a\,\sign (k) e^{\frac{a^2}{2}s}
\Nc\bigg(\sign(k)\frac{as-k}{\sqrt{s}}\bigg)
\text{ for }  k\in\R,
\end{equation*}
where
$$\sign (k)=\begin{cases}
1& \text{ for  } k\geq 0,\\
-1& \text{ for  } k<0.
\end{cases}
$$
Define 
\begin{equation}
\label{F}
F\left(T, a, k\right)=\int\limits_0^T\phi(T-s)\psi(a,s,k)e^{-\frac{a^2}{2}s}ds,
\end{equation}
where  $\phi$ is the function defined in~\eqref{phi}.

Consider European options with  strike $K$ and  time to expiry $T$. 
Let  $\C(S_0, K, T)$  and  $\Put(S_0, K, T)$ be the price of  the call option and 
the  put option respectively, when the initial value 
of the underlying price is  $S_0$.

\begin{theorem}
\label{Call-Put}
If  $S_0=1$ and $K>1$, then 
\begin{equation*}
\C(1, K, T)=2p\bigg(F\Big(T,\frac{\sgp}{2}, k\Big)-e^{\sgp k}
F\Big(T, -\frac{\sgp}{2},k\Big)\bigg),
\end{equation*}
where  $k=\log K/\sgp$.

If  $S_0=1$ and $K<1$, then 
\begin{equation*}
\Put(1, K, T)=2q\bigg(e^{\sgm k}F\Big(T,-\frac{\sgm}{2}, k\Big)-
F\Big(T, \frac{\sgm}{2},k\Big)\bigg),
\end{equation*}
where  $k=\log K/\sgm$.
\end{theorem}

By Theorem~\ref{Call-Put} we immediately get the following equation for 
 the ATM price
\begin{equation}
\label{call-put-atm}
\begin{split}
\Vatm(T)&=\C(1,1,T)=2p\bigg(F\Big(T,\frac{\sgp}{2}, 0\Big)-
F\Big(T, -\frac{\sgp}{2},0\Big)\bigg)\\
&=\Put(1,1,T)
=2q\bigg(F\Big(T,-\frac{\sgm}{2}, 0\Big)-F\Big(T, \frac{\sgm}{2},0\Big)\bigg),
\end{split}
\end{equation}
which still involves the function $F$.
Corollary ~\ref{ATM-price} below shows that the above equation for 
 the ATM price can be simplified in such a way that 
  the ATM price is analytically  expressed in terms of the error function
$${\rm Erf}(x)=\frac{2}{\sqrt{\pi}}\int\limits_0^xe^{-t^2}dt.$$
\begin{corollary}[{\bf ATM price}]
\label{ATM-price}
If $S_0=K=1$, then 
\begin{equation*}
\begin{split}
\Vatm(T)
&=
\frac{\sgm^2\sgp^2}{4(\sgm^2-\sgp^2)}\bigg(\frac{\sqrt{8T}}{\sgp\sqrt{\pi}}e^{-\frac{\sgp^2}{8}T}-
\frac{\sqrt{8T}}{\sgm\sqrt{\pi}}e^{-\frac{\sgm^2}{8}T}\bigg.\\
&
\bigg.+\Big(\frac{4}{\sgp^2}+T\Big){\rm Erf}\Big(\frac{\sgp\sqrt{T}}{\sqrt{8}}\Big)
-\Big(\frac{4}{\sgm^2}+T\Big){\rm Erf}\Big(\frac{\sgm\sqrt{T}}{\sqrt{8}}\Big)\bigg).
\end{split}
\end{equation*}
\end{corollary}

Proofs of both Theorem~\ref{Call-Put} and Corollary~\ref{ATM-price} are given in Section~\ref{A}.

\begin{remark}
{\rm  It should be  noted  that the  option pricing formulas  in Theorem~\ref{Call-Put}
differ from those that are given  in Pigato~\cite{Pigato}. In both papers 
the option price is given by a single integral, but the corresponding integrands
differ. It might be of interest to investigate the relationship between the two variants. At the same time, the ATM price in Corollary~\ref{ATM-price}
 is exactly the same as  the one  in~\cite{Pigato}.
In addition, note that the ATM price can be rewritten as follows
$$\Vatm(T)=\frac{\sgm^2\sgp^2}{4(\sgm^2-\sgp^2)}
\big(I(\sgp, T)-I(\sgm, T)\big),$$
where 
\begin{equation*}
I(x, T)=\frac{\sqrt{8T}}{x\sqrt{\pi}}e^{-\frac{x^2}{8}T}
+\Big(\frac{4}{x^2}+T\Big){\rm Erf}\Big(\frac{x\sqrt{T}}{\sqrt{8}}\Big).
\end{equation*}
}
\end{remark}

\subsection{Option prices and Dupire's forward equation}


In this section we provide (in Theorem~\ref{Dupire_rep} below) another
representation of option prices. This representation 
gives the price of an in-the-money option in terms 
of a weighted integral of the corresponding ATM price over the time 
until maturity.
It is  based on the well known 
Dupire's forward equation (Dupire~\cite{Dupire},~\cite{Dupire2}), 
which we recall below. 

If the underlying  price follows the local volatility model
$$dS_t=\sigma(t, S_t)S_tdW_t,$$ 
then  the  price $\C(K,T)$ of a European call option with strike $K$ and time to expiry
$T$ satisfies  the equation (the forward equation)
\begin{equation*}
\frac{\partial \C(K, T)}{\partial T}=\frac{1}{2}K^2\sigma^2(T,K)\frac{\partial^2 \C(K,T)}{\partial K^2}.
\end{equation*}
It follows from  the forward  equation that  
\begin{align*}
	\C\left(K,T\right)-\left(S_0-K\right)_+
	&=\int\limits_0^T\frac{\partial}{\partial t}\C\left(K,t\right)dt \\
	&=\frac{K^2}{2}\int\limits_0^T\sigma^2(t,K)\frac{\partial^2}{\partial K^2}\C\left(K,t\right)dt\\
	&=\frac{K^2}{2}\int\limits_0^T\sigma^2(t,K)\frac{\partial^2}{\partial K^2}
	\E[\max(S_t-K,0)|S_0]dt\\
	&=\frac{K^2}{2}\int\limits_0^T\sigma^2 (t,K)\E[\delta (S_t-K)|S_0]dt,
\end{align*}
where $\delta(\cdot)$ is the delta-function.
Noting  that $\E[\delta(S_t-K)|S_0]=p_S(S_0,K,t)$,
 where $p_S(S_0,\cdot,t)$, is the probability density function of $S_t$ given 
$S_0$, we arrive to the equation
\begin{equation}
\label{Dup1}
\C(K,T)-(S_0-K)_+=\frac{K^2}{2}\int\limits_0^T\sigma^2(t,K)p_S(S_0,K, t)dt,
\end{equation}
which we are going to use in the proof of Theorem~\ref{Dupire_rep} below.
\begin{theorem}
\label{Dupire_rep}
If  $S_0=1$ and $K>1$, then 
\begin{equation*}
  \C(1,K,T)=\sqrt{K}\int\limits_0^T\Vatm(T-s)h(s,\log(K)/\sgp)
e^{-\frac{1}{8} \sgp^2 s}ds;
\end{equation*}
if  $S_0=1$ and $K<1$, then
\begin{equation*}
\Put(1,K,T)=\frac{1}{\sqrt{K}}\int\limits_0^T\Vatm(T-s)h(s,-\log(K)/\sgm)
e^{-\frac{1}{8} \sgm^2 s}ds,
\end{equation*}
where in both cases $\Vatm$ is the ATM price.
\end{theorem}
The proof of Theorem~\ref{Dupire_rep} is given  in Section~\ref{A}.

\subsection{Black-Scholes approximation}
\label{BS}

In this section we discuss  the approximation of  option prices in the two-valued LVM 
model (i.e. $\C(1, K, T)$ and $\Put(1, K, T)$)  by the corresponding BS 
prices.

Denote by 
$\CBS(\sigma, S_0, K, T)$ and  $\PBS(\sigma, S_0, K, T)$ the  BS prices of 
 European call option and European put option respectively  with 
 strike $K$ and time to maturity $T$,
  when  volatility  of the underlying asset is equal to $\sigma$.
\begin{theorem}
\label{BS-approx}
If  $K>1$, then 
$$|\C(1, K, T)-2p\CBS(\sgp, 1, K, T)|\leq cT,$$
and if $K<1$, then 
$$|\Put(1, K, T)-2q\PBS(\sgm, 1, K, T)|\leq cT$$
for some constant $c$ and all sufficiently small $T$.
\end{theorem}

The proof of Theorem~\ref{BS-approx} is given in Section~\ref{BS-proof}.

\begin{remark}
{\rm The BS approximation of option prices in Theorem~\ref{BS-approx}
is an important tool in our analysis of implied volatility in  Section~\ref{Volatility}.
Another application of the 
 approximation is given in Section~\ref{BS-study},
where it is used for estimating implied volatility.

  It should be also noted that the BS approximation is not immediately visible from the final pricing formulas.  However, it is
   readily obtained, if the option price is written in terms of an integral of the  joint density of the underlying price process  and  its functionals. 
}
\end{remark}

\section{Implied volatility}
\label{Volatility}

In this section we use the BS approximation to obtain 
some  results from~\cite{Pigato} concerning implied volatility. 
Recall that implied volatility 
  $\sigmaBS(T, k)$ is  considered 
as a function of maturity $T$ and log-moneyness $k=\log(K/S_0)$.
Note that we have $k=\log K$  in the case when  $S_0=1$. 

Start with a remark that is almost verbatim to Remark 3.4 in~\cite{Pigato}
concerning the ATM implied volatility $\sigatm(T):=\sigmaBS(T,0)$.
By definition, $\sigatm(T)$ is  the solution of the equation
\begin{equation*}
\CBS(\sigatm(T), 1, 1, T)=\Vatm(T).
\end{equation*}
Recall that given volatility $\sigma$ we have that 
\begin{equation}
\label{BS-ATM}
\begin{split}
\BS(\sigma, T)&:=
\CBS(\sigma, 1, 1, T)
=\Nc\Big(\frac{\sigma\sqrt{T}}{2}\Big)-\Nc\Big(-\frac{\sigma\sqrt{T}}{2}\Big)
={\rm Erf}\Big(\frac{\sigma\sqrt{T}}{\sqrt{8}}\Big).
\end{split}
\end{equation}
Therefore,
$$
\sigatm(T)=\frac{\sqrt{8}}{\sqrt{T}}{\rm Erf}^{-1}\big(\Vatm(T)\big).
$$
Further, using that
$${\rm Erf}^{-1}(x)=
\frac{\sqrt{\pi}}{2}\Big(x+\frac{\pi}{12}x^3\Big)+o(x^3), \text{ as }  x\to 0,$$
we obtain  the short term expansion for the ATM implied volatility 
\begin{equation*}
\sigatm\left(T\right)=2 \frac{\sgm\sgp}{ \left(\sgm+\sgp\right)}-\frac{\left(\sgm\sgp\right)^2
\left(\sgm-\sgp\right)^2}{12 \left(\sgm+\sgp\right)^3}T+o(T), \text{ as } T\to 0.
\end{equation*}
In addition, note that 
\begin{equation*}
\sigatm:=\sigatm\left(0\right)=2 \frac{\sgm\sgp}{ \left(\sgm+\sgp\right)}
=2p\sgp=2q\sgm.
\end{equation*}

\subsection{Implied volatility in the central limit regime}

In this section we consider the short term  behaviour 
of implied volatility  in the case when strike and maturity 
are related by the equation $K=e^{\gamma\sqrt{T}}$. 
This case was considered in~\cite[Theorem 3.1]{Pigato},
where it was called  the central limit regime (and we use  the same  terminology).

One of the results of the aforementioned theorem is the equation
\begin{equation}
\label{iv-eq}
\begin{split}
\lim_{T\to 0}\sigmaBS(T, \gamma\sqrt{T}) 
&=\frac{2 \sgm\sgp}{\sgm+\sgp}+
\frac{\sqrt{\pi}}{\sqrt{2}}\frac{ \sgp-\sgm}{\sgm+\sgp}\gamma\\
&+ 
\frac{\sgp-\sgm}{ \sgp\sgm}\Big(\frac{\sgp-\sgm}{ 2\left(\sgm+\sgp\right)}-
\sign(\gamma)\Big)\gamma^2
 + o(\gamma^2), \text{ as }\gamma\to 0.
\end{split}
\end{equation}
In other words,  the implied volatility  
$\sigmaBS(T, \gamma\sqrt{T})$ can be
 approximated for short term maturities $T$  by a quadratic  function of $\gamma$,  
 which can be computed explicitly.  Below we  compute this quadratic function 
 by using the BS approximation derived in Theorem~\ref{BS-approx}.

For definiteness assume that $\gamma>0$ (i.e. $K=e^{\gamma\sqrt{T}}>1$) and 
use our result for prices of call options.
Let $\C(1,e^{\gamma\sqrt{T}},T)$ be  the call price  in the two-valued LVM
 (Theorem~\ref{Call-Put}).
The equation for the  implied volatility $\sigmaBS$  is 
\begin{equation}
\label{ivol0}
\CBS\big(\sigmaBS, 1, e^{\gamma\sqrt{T}}, T\big)=\C\big(1,e^{\gamma\sqrt{T}},T\big),
\end{equation} 
where the left hand side is 
the BS price of the call option with maturity $T$ and the strike $K=e^{\gamma\sqrt{T}}$. 
Then
$$\CBS\big(\sigma, 1, e^{\gamma\sqrt{T}}, T\big)=\Nc(d_1)-
e^{\gamma\sqrt{T}}\Nc(d_0),$$
where 
$d_1=-\frac{\gamma}{\sigma}+\frac{\sigma\sqrt{T}}{2}$  and $d_0=-\frac{\gamma}{\sigma}-\frac{\sigma\sqrt{T}}{2}$.
By Taylor's theorem we have  that
\begin{align}
\label{BS-Taylor}
\CBS(\sigma, 1, e^{\gamma\sqrt{T}}, T)
&= \bigg(\frac{\sigma}{\sqrt{2\pi}}-\frac{\gamma}{2}+\frac{\gamma^2}{\sigma\sqrt{2\pi}}\bigg)
\sqrt{T}+o\big(\sqrt{T}\big), \text{ as }  T\to 0.
\end{align}
Further, combining the BS approximation (Theorem~\ref{BS-approx}) for the right hand side of~\eqref{ivol0}
with~\eqref{BS-Taylor} gives that
\begin{equation}
\label{Loc-Taylor}
\begin{split}
\C\big(1,e^{\gamma\sqrt{T}},T\big)&=2p\CBS\big(\sgp, 1, e^{\gamma\sqrt{T}}, T\big)\\
&=2p\bigg(\frac{\sgp}{\sqrt{2\pi}}-\frac{\gamma}{2}+
\frac{\gamma^2}{\sgp\sqrt{2\pi}}\bigg)\sqrt{T}+o\big(\sqrt{T}\big).
\end{split}
\end{equation} 
Replacing   both the left hand and the right hand sides of equation~\eqref{ivol0} by their  approximations
(provided by equations~\eqref{BS-Taylor} and~\eqref{Loc-Taylor} respectively) we obtain the equation
\begin{equation*}
\bigg(\frac{\sigma}{\sqrt{2\pi}}-\frac{\gamma}{2}+
\frac{\gamma^2}{\sigma\sqrt{2\pi}}\bigg)
\sqrt{T}=2p\bigg(\frac{\sgp}{\sqrt{2\pi}}-\frac{\gamma}{2}+
\frac{\gamma^2}{\sgp\sqrt{2\pi}}\bigg)\sqrt{T}+o\big(\sqrt{T}\big),
\end{equation*}
which means that under the assumptions made  the implied volatility 
$\sigmaBS(\gamma, e^{\gamma\sqrt{T}})$ 
converges to a limit, as $T\to 0$, and, moreover, this limit 
can be estimated by the solution of the equation
\begin{equation}
\label{ivol1}
\frac{\sigma}{\sqrt{2\pi}}-\frac{\gamma}{2}+\frac{\gamma^2}{\sigma\sqrt{2\pi}}
=p\sgp\sqrt{\frac{2}{\pi}}-p\gamma+\frac{p}{\sgp}\sqrt{\frac{2}{\pi}}\gamma^2.
\end{equation}
It is easy to see that equation~\eqref{ivol1} is basically 
a quadratic equation for the unknown $\sigma$ with coefficients analytically depending on $\gamma$.
 This implies that the solution is 
 an analytic  function $\sigma(\gamma)$ of $\gamma$.
Consider  Taylor's expansion of the second order for 
this function at $\gamma=0$, that is 
\begin{equation*}
\sigma(\gamma)=
\bigg(c_0+c_1\gamma +\frac{1}{2}c_2\gamma ^2\bigg)+o\big(\gamma ^2\big),
\text{ as } \gamma\to 0,
\end{equation*}
where   $c_0$, $c_1$ and $c_2$ denote the values at $0$ of the function itself,  
its 1st and 2nd  derivatives respectively.
Using this expansion for approximating the left hand side of~\eqref{ivol1}
gives the equation
\begin{align*}
\frac{c_0+c_1\gamma +\frac{1}{2}c_2\gamma ^2}{\sqrt{2\pi}}&-
\frac{\gamma}{2}+\frac{\gamma^2}{\sqrt{2\pi}}
\frac{1}{(c_0+c_1\gamma +\frac{1}{2}c_2\gamma ^2)}\\
&=\frac{c_0}{\sqrt{2\pi}}+\left(\frac{c_1}{\sqrt{2\pi}}-\frac{1}{2}\right)\gamma+
\frac{1}{\sqrt{2\pi}}\left(\frac{c_2}{2}+\frac{1}{c_0}\right)\gamma^2+o(\gamma^2),
\end{align*}
which, in turn, implies that 
\begin{equation}
\label{ivol3}
\begin{split}
\frac{c_0}{\sqrt{2\pi}}+\bigg(\frac{c_1}{\sqrt{2\pi}}-\frac{1}{2}\bigg)\gamma+
\frac{1}{\sqrt{2\pi}}\bigg(\frac{\sigma_2}{2}+\frac{1}{c_0}\bigg)\gamma^2
&
=p\sgp\frac{\sqrt{2}}{\sqrt{\pi}}-p\gamma+\frac{p}{\sgp}\frac{\sqrt{2}}{\sqrt{\pi}}\gamma^2+o\big(\gamma^2\big).
\end{split}
\end{equation}
Equating coefficients at $\gamma^i,\, i=0,1,2$ in~\eqref{ivol3}
we obtain that 
$$c_0=2p\sgp,  c_1=\frac{\sqrt{2}}{\sqrt{\pi}}(1-2p) \text{ and } 
c_2=\frac{4p^2-1}{p\sgp},$$
and, hence,
\begin{equation}
\label{imp-gamma-call}
\begin{split}
\sigma(\gamma)&=2p\sgp +\frac{\sqrt{2}}{\sqrt{\pi}}(1-2p)\gamma+\frac{1}{2}\frac{4p^2-1}{p\sgp}\gamma^2
+o(\gamma^2)\\
&=
\frac{2 \sgm\sgp}{\sgm+\sgp}+
\frac{\sqrt{\pi}}{\sqrt{2}}\frac{ \sgp-\sgm}{\sgm+\sgp}\gamma + 
\frac{\sgp-\sgm}{ \sgp\sgm}\bigg(\frac{\sgp-\sgm}{ 2(\sgm+\sgp)}-
1\bigg)\gamma^2+o\big(\gamma^2\big).
\end{split}
\end{equation}
Alternatively, one can use the put price  and  repeat the above argument 
in the case when $\gamma<0$, i.e. $K=e^{\gamma\sqrt{T}}<1$,  to obtain that
\begin{equation}
\label{imp-gamma-put}
\begin{split}
\sigma(\gamma)&=
2q\sgm+\frac{\sqrt{2}}{\sqrt{\pi}}(2q-1)\gamma +\frac{1}{2}\frac{4q^2-1}{q\sgm}\gamma^2
+o\big(\gamma^2\big)\\
&=
\frac{2 \sgm\sgp}{\sgm+\sgp}+
\frac{\sqrt{\pi}}{\sqrt{2}}\frac{ \sgp-\sgm}{\sgm+\sgp}\gamma + 
\frac{\sgp-\sgm}{ \sgp\sgm}\bigg(\frac{\sgp-\sgm}{ 2(\sgm+\sgp)}+
1\bigg)\gamma^2+o\big(\gamma^2\big).
\end{split}
\end{equation}
Finally, note  that~\eqref{imp-gamma-call} and~\eqref{imp-gamma-put} are special cases of~\eqref{iv-eq}
depending on the sign of $\gamma$.

\subsection{ATM implied volatility skew}

Recall that
$\CBS(\sigma,1,e^k,T)$ 
is the  BS price (with  volatility $\sigma$) of the call option with  the 
log-strike $k=\log K$ and maturity $T$. As before,  let $\C(1, e^k, T)$ be the call price of the same option  in the two-valued LVM.
Given the log-strike $k$ and maturity $T$ the corresponding  implied volatility
 $\sigmaBS(T,k)$  is  defined as the solution of the  equation 
$\CBS(\sigma, 1, e^k, T)=\C(1, e^k, T)$ for $\sigma$.

Denote $\partial_k=\frac{\partial }{\partial k}$.

\begin{theorem}
\label{ATM_Skew}
The ATM skew is given by 
\begin{equation}
\label{sigmaBS(T,0)}
\begin{split}
\partial_k\sigmaBS(T,0)
&=\frac{\sqrt{\pi}}{\sqrt{2T}}e^{\frac{1}{8}\sigmaBS^2(T,0)T} 
\bigg(1-\frac{2\sgm}{\sgm+\sgp}\Big(
F\big(T,-\sgp/2,0\big)+F\big(T,\sgp/2,0\big)\Big)\bigg)
\end{split}
\end{equation}
and
\begin{equation}
\label{skewATM}
\partial_k \sigmaBS(T,0)
=
\frac{\sqrt{\pi}}{\sqrt{2T}}
\frac{\sgp-\sgm}{\sgm+\sgp}+o\big(\sqrt{T}\big), \text{ as } T\to 0.
\end{equation}
\end{theorem}

\begin{remark}
{\rm 
Equation~\eqref{sigmaBS(T,0)}  is similar to the equation  for the ATM skew obtained in
Theorem 3.5 in~\cite{Pigato}. In particular, the factor 
$\frac{\sqrt{\pi}}{\sqrt{2T}}e^{\frac{1}{8}\sigmaBS^2(T,0)T}$
is exactly the same as the one in that theorem. However, the term
$F(T,-\sgp/2,0)+F(T,\sgp/2,0)$ differs from the similar term 
in~\cite{Pigato}. This difference  is caused by the same reason as the  
difference in the pricing formulas. 
Note that the short term asymptotic behaviour of the ATM skew given by~\eqref{skewATM} 
is exactly the same as in~\cite{Pigato} (e.g. see Remark 3.2 in that paper).
}
\end{remark}

Before proceeding to  the proof of Theorem~\ref{ATM_Skew}, 
we prove below  two auxiliary statements.
The first one is Proposition~\ref{F-asymp}  
that provides an asymptotic result for the function 
$F$ defined in~\eqref{F}.
The second auxiliary statement is 
Proposition~\ref{log-strike-derivative} that 
concerns  the derivative of the call price  with respect to the  log-strike.

\begin{proposition}
\label{F-asymp}
For any fixed $a\in \R$ the following holds
\begin{equation*}
F(T,a,0)=\frac{a}{\sqrt{2\pi}}\sqrt{T}+\frac{1}{2}+o\big(\sqrt{T}\big),
 \text{ as } T\to 0.
\end{equation*}
\end{proposition}
\begin{proof}
Observe  that 
\begin{align*}
\psi(a,t,0)&=\frac{1}{\sqrt{2\pi t}}+\frac{a}{2}+O\big(\sqrt{t}\big),
 \text{ as } t\to 0,\\
\phi(t)&=\frac{1}{\sqrt{2\pi t}}+\bigg(\frac{\sgp+\sgm}{8\sqrt{2\pi}}-\frac{\sgp\sgm}{4}\bigg)\sqrt{t}
+o\big(\sqrt{t}\big), \text{ as } t\to 0.
\end{align*}
Therefore,
\begin{align*}
F\left(T,a,0\right)&=
\frac{1}{\sqrt{2\pi}}\int_0^T\frac{e^{-\frac{a^2}{2}s}}{\sqrt{T-s}}
\bigg(\frac{a}{2}
+\frac{1}{\sqrt{2\pi s}}\bigg)ds+o\big(\sqrt{T}\big)\\
&=\frac{a}{2\sqrt{2\pi}}\int_0^T\frac{1}{\sqrt{T-s}}ds+\frac{1}{2\pi}\int_0^T
\frac{1}{\sqrt{s(T-s)}}ds+o\big(\sqrt{T}\big)\\
&
=\frac{a}{\sqrt{2\pi}}\sqrt{T}+\frac{1}{2}+o\big(\sqrt{T}\big), \text{ as } T\to 0.
\end{align*}
The proof of the proposition is finished.
\end{proof}
\begin{proposition}
\label{log-strike-derivative}
We have that 
\begin{equation}
\label{Ck}
\partial_k\C(1,e^{k},T)=
-2\frac{\sgm}{\sgm+\sgp}F\big(T,-\sgp/2,k/\sgp\big)e^{k},
\end{equation}
and
\begin{equation}
\label{Ck0}
\begin{split}
\partial_k\C(1,1,T)&:=\partial_k\C\big(1,e^{k},T\big)|_{k=0}=-2\frac{\sgm}{\sgm+\sgp}
F\big(T,-\sgp/2,0\big)\\
&
=-\frac{\sgm}{\sgm+\sgp} +
\frac{1}{\sqrt{2\pi}}\frac{\sgm\sgp}{\sgm+\sgp}\sqrt{T}+o\big(\sqrt{T}\big), 
\text{ as } T\to 0.
\end{split}
\end{equation}
\end{proposition}
\begin{proof}
By Theorem~\ref{Call-Put} we have that 
$$
\C\big(1, e^{k}, T\big)=2p\Big(F\big(T,\sgp/2, k/\sgp\big)-e^{k}
F\big(T, -\sgp/2, k/\sgp\big)\Big),
$$
where
$$F\big(T, a, k/\sgp\big)=\int_0^T\phi(T-s)\psi\big(a,s, k/\sgp\big)
e^{-\frac{a^2}{2}s}ds,$$
and functions $\phi$ and $\psi$ are given by~\eqref{phi} and~\eqref{psi} respectively.
Further, observe that 
$\partial_k\psi(a,s,k)=0$ for all $a, s$ and $k$.
Therefore,
\begin{align*}
 \partial_k\C\big(1,e^{k},T\big)
&=-2p e^{k}\int_0^T\phi(T-s)\psi\big(-\sgp/2,s, k/\sgp\big)e^{-\frac{\sgp^2}{8}s}ds\\
&=
-2\frac{\sgm}{\sgm+\sgp}F\big(T,-\sgp/2,k/\sgp\big)e^{k},
\end{align*}
as claimed in~\eqref{Ck}, and, hence, 
we get that
$$\partial_k\C(1,1,T)=-2\frac{\sgm}{\sgm+\sgp}F\big(T,-\sgp/2,0\big),$$
i.e. the first equation in~\eqref{Ck0}. 
Applying Proposition~\ref{F-asymp} with $a=-\frac{\sgp}{2}$ gives 
the second equation in~\eqref{Ck0},  the short term asymptotics of 
$\partial_k\C(1,1,T)$, as claimed.
\end{proof}

\begin{proof}[Proof of Theorem~\ref{ATM_Skew}]
Similarly to $\partial_k=\frac{\partial }{\partial k}$,
denote $\partial_{\sigma}=\frac{\partial }{\partial \sigma}$.
By the chain rule, we have that 
\begin{equation}
\label{sigmaBS_k}
\partial_k\sigmaBS(T,k)=
\frac{\partial_k\C\big(1,e^k,T\big)-\partial_k\CBS\big(\sigmaBS(T,k), 1, e^k,T\big)}
{\partial_{\sigma}\CBS\big(\sigmaBS(T,k),1, e^k,T\big)}.
\end{equation}
Further, observe that 
\begin{align}
\label{Vega}
\partial_{\sigma}\CBS(\sigmaBS(T,0),1,1,T)&=\frac{\sqrt{T}}{\sqrt{2\pi}}
e^{-\frac{1}{8}\sigmaBS^2(T,0)T},
\\
\label{Ck(0)}
\partial_k\CBS(\sigmaBS(T,0),1,1,T)&=
\frac{1}{2}\bigg(-1+{\rm Erf}\Big(\frac{\sqrt{T} \sigmaBS(T,0)}{2 \sqrt{2}}\Big)\bigg),
\end{align}
and 
\begin{equation}
\label{BS=atm}
{\rm Erf}\Big(\frac{\sqrt{T}\sigmaBS(T,0)}{2 \sqrt{2}}\Big)=
\CBS(\sigmaBS(T,0),T)=\Vatm(T).
\end{equation}
Using~\eqref{Vega},~\eqref{Ck(0)} and~\eqref{BS=atm} we can rewrite~\eqref{sigmaBS_k} in the case $k=0$ as follows
\begin{equation}
\label{skewT}
\partial_k\sigmaBS(T,0)=\frac{\sqrt{2 \pi} }{\sqrt{T}}
e^{\frac{1}{8}\sigmaBS^2(0,T)T} 
\bigg(\partial_k\C(1,1,T)+\frac{1}{2}-\frac{\Vatm(T)}{2}\bigg)
\end{equation}
Next, by equation~\eqref{call-put-atm},
$$\Vatm(T)=
\frac{2\sgm}{\sgm+\sgp}
\Big(F\big(T,\sgp/2,0\big)-F\big(T,-\sgp/2,0\big)\Big)
$$
and, by Proposition~\ref{log-strike-derivative}, 
$$\partial_k\C(1,1,T)=-\frac{2\sgm}{\sgm+\sgp}F\big(T,-\sgp/2,0\big).$$
Therefore,
$$
\partial_k\C(1,1,T)+\frac{1}{2}-\frac{\Vatm(T)}{2}
=\frac{1}{2}-\frac{\sgm}{\sgm+\sgp}\Big(
F\big(T,-\sgp/2,0\big)+F\big(T,\sgp/2,0\big)\Big)
$$
and, getting back to~\eqref{skewT}, we obtain, after simple algebra, that 
$\sigmaBS(T,0)$ is equal to~\eqref{sigmaBS(T,0)},
as claimed.

Finally, equation~\eqref{skewATM}
follows from~\eqref{sigmaBS(T,0)} and Proposition~\ref{F-asymp} (we skip details).

\end{proof}

\section{Proofs of Theorems~\ref{L0}--\ref{BS-approx}
and Corollary~\ref{ATM-price}}
\label{A}


\subsection{Proof of Theorem~\ref{L0}}

By equation~\eqref{f-v-u-s-1} we have  that
\begin{equation}
\label{density}
p(0,x,T)
=2\alpha(x)e^{-\frac{\sigma(x)x}{2}}Q(x,T),
\end{equation}
where 
$$Q(x,T)=\int\limits_{u+v+s=T}
\Big(\int_0^{\infty}h(v, \ell p)h(u, \ell q)d\ell\Big)h(s, x)
e^{-\frac{\sgp^2v+\sgm^2u+\sigma^2(x)s}{8}}dsdv.
$$
Recall two equations 
 that were used  in the proof of~\cite[Theorem 3, Part 1)]{GS}, namely
\begin{equation*}
\int\limits_0^{\infty }h(v,\ell p)h(u,\ell q)dl=\frac{pq}{2\sqrt{2 \pi } 
\left(p^2u+q^2v\right)^{3/2}},
\end{equation*}
and 
\begin{equation*}
\int\limits_0^w
\frac{pq}{\sqrt{2 \pi }\left(p^2(w-v)+q^2v\right)^{3/2}}e^{-\frac{\sigma_{\rm p}^2v+\sigma^2_{\rm m}(w-v)}{8}}dv
=\phi(w) \text{ for } w>0,
\end{equation*}
where the function $\phi$ is defined in~\eqref{phi}.
Using these equations in~\eqref{density} gives  the claimed equation for the  density 
of $X_T$.

\subsection{Proof of Theorem~\ref{Call-Put}}

Since $S_0=1$ and $K>1$, we have that  $X_0=\log S_0/\sgp=0$
and $k=\log K/\sgp\geq 0$ respectively. By Theorem~\ref{L0} we have that 
\begin{align*}
\C(1, K, T)&=\int_{k}^{\infty}\Big(e^{\sgp x}-e^{\sgp k}\Big)p(0,x,T)dx\\
&=2p\int_{k}^{\infty}\Big(e^{\sgp x}-e^{\sgp k}\Big)
e^{-\frac{\sgp}{2}x}\Big(\int_0^T\phi(T-s)h(s,x)e^{-\frac{\sgp^2}{8}s}ds\Big)dx\\
&=2p\int_0^T\phi(T-s)e^{-\frac{\sgp^2}{2}s}\bigg(\int_{k}^{\infty}
\Big(e^{\frac{\sgp}{2}x} - e^{\sgp k}e^{-\frac{\sgp}{2}x}\Big)
h(s,x)dx\bigg)ds\\
&=2p\int_0^T\phi(T-s)e^{-\frac{\sgp^2}{8}s}\Big(
\psi\big(\sgp/2,s,k\big)
-e^{\sgp k}\psi\big(-\sgp/2,s,k\big)\Big)ds\\
&=2p\Big(F\big(T, \sgp/2, k\big)-e^{\sgp k}
F\big(T, -\sgp/2, k\big)\Big),
\end{align*}
as claimed.

Equation for the put price $\Put(1, K, T)$ can be obtained 
similarly.

\subsection{Proof of Corollary~\ref{ATM-price}}
Note  first  that 
\begin{align}
\label{psi(k=0)}
\psi(a,s,0)-\psi(-a,s,0)&=ae^{\frac{a^2}{2}s} \text{ for all }  a.
\end{align}
Combining~\eqref{psi(k=0)} with~\eqref{call-put-atm} we obtain  that 
\begin{align*}
\Vatm(T)&=\C(1,1,T)\\
&=2p\int\limits_0^T\phi(T-s)e^{-\frac{\sgp^2}{8}s}
\Big(\psi\big(\sgp/2,s,0\big)-
\psi\big(-\sgp/2,s,0\big)\Big)ds\\
&=p\sgp\int\limits_0^T\phi(T-s)ds=p\sgp\int\limits_0^T\phi(s)ds
=\frac{\sgm\sgp}{\sgm+\sgp}\int\limits_0^T\phi(s)ds.
\end{align*}
Further, a direct computation gives  that
\begin{equation*}
\begin{split}
\int_0^T\phi(T-s)ds&=
\frac{\sqrt{T}}{\sqrt{2\pi}}\frac{1}{(\sgm-\sgp)}\Big(\sgm e^{-\frac{\sgp^2}{8}T}-
\sgp e^{-\frac{\sgm^2}{8}T}\Big)\\
&+\frac{\sgm(4+\sgp^2T)}{4\sgp(\sgm-\sgp)}{\rm Erf}
\bigg(\sgp\frac{\sqrt{T}}{\sqrt{8}}\bigg)
-\frac{\sgp(4+\sgm^2T)}{4\sgm(\sgm-\sgp)}{\rm Erf}
\bigg(\sgm\frac{\sqrt{T}}{\sqrt{8}}\bigg).
\end{split}
\end{equation*}
Combining  the above and simplifying gives the ATM price, as claimed.

\subsection{Proof of Theorem~\ref{Dupire_rep}}

Recall that in the two-valued LVM
$\sigma(t, K)=\sgp{\bf 1}_{\{K\geq 1\}}+\sgm{\bf 1}_{\{K<1\}}$.
If  $S_0=1$, then  equation~\eqref{Dup1} becomes 
\begin{align*}
\C\left(1,K,T\right)
&=\frac{K^2\sgp^2}{2}\int\limits_0^Tp_S\left(S_0,K,t\right)dt \text{ for } K>1.
\end{align*}
 Noting that 
$$p_S\left(1,K,t\right)=\frac{p(0,\log K/\sgp,t)}{K\sgp}$$ 
for $K>1$, where $p(0,\cdot,t)$ is the density of $X_t$ given that $X_0=0$
(see Lemma~\ref{L0}), we get   that
\begin{align*}
\C(1,K, T)&=\frac{K\sgp}{2}\int\limits_0^T
p(0,\log K/\sgp,t)dt\\
&=\sqrt{K}p\sgp\int\limits_0^T\int\limits_0^t\phi(t-s)h\big(s,\log K/\sgp\big)
e^{-\frac{\sgp^2}{8}s}dsdt\\
&=\sqrt{K}
\int\limits_0^Th\big(s,\log K/\sgp\big)e^{-\frac{\sgp^2}{8}s}
\Big(p\sgp\int_0^{T-s}\phi(t)dt\Big)ds\\
&=\sqrt{K}\int\limits_0^Th\big(s,\log K/\sgp\big)e^{-\frac{\sgp^2}{8}s}\Vatm(T-s)ds
\text{ for }  K>1,
\end{align*}
where $\Vatm(T-s)$ is the ATM price (Theorem~\ref{ATM-price}) for maturity $T-s$.

\subsection{Proof of Theorem~\ref{BS-approx}}
\label{BS-proof}

We obtain the  BS approximation only for the  call option price $\C(1, K, T)$, as the case of the put option is similar.

Using equation~\eqref{density} for the density of $X_T$ 
gives that
$$\C(1, K, T)=2p\int\limits_k^{\infty}\Big(e^{\sgp x}-e^{\sgp k}\Big)e^{-\frac{\sgp}{2}x}
A(x,\sgp,\sgm)dx,$$
where 
\begin{align*}
A(x,\sgp,\sgm)&=
\int\limits_0^{\infty}\int\limits_0^T\int\limits_{0}^{T-s}h(v,\ell p)
h(T-v-s,\ell q)h(s,x)
e^{-\frac{\sgp^2v+\sgm^2(T-v-s)+\sgp^2 s}{8}}dvdsd\ell.
\end{align*}
For the BS price of  the same call option in the case when  volatility is equal to 
$\sgp$   we have that 
$$
\CBS(\sgp, 1, K, T)
=\int\limits_k^{\infty}\Big(e^{\sgp x}-e^{\sgp k}\Big)e^{-\frac{\sgp}{2}x}
A(x,\sgp,\sgp)dx
$$
Denote $u=T-v-s$ and observe that 
\begin{align*}
\Big|e^{-\frac{\sgp^2v+\sgm^2 u+\sgp^2 s}{8}} - e^{-\frac{\sgp^2T}{8}}\Big|&=
\Big|1-e^{-\frac{1}{8}(\sgm^2-\sgp^2)u}\Big|e^{-\frac{\sgp^2T}{8}}.
\end{align*}
Therefore, 
\begin{align*}
\Big|e^{-\frac{\sgp^2v+\sgm^2 u+\sgp^2 s}{8}} - e^{-\frac{\sgp^2T}{8}}\Big|
&
\leq \frac{|\sgm^2-\sgp^2|T}{8}e^{-\frac{\sgp^2}{8}T}+o(T), 
\text{ as }T\to 0.
\end{align*}
This gives that 
$$\bigg|\int\limits_k^{\infty}\Big(e^{\sgp x}-e^{\sgp k}\Big)
e^{-\frac{\sgp}{2}x}
\big(A(x,\sgp,\sgm)-A(x,\sgp, \sgp)\big)dx\bigg|\leq C_1T$$
for some constant $C_1$  and all sufficiently small $T$, which in turn 
 implies that
$$|\C(1,K,T)-2p\CBS(\sgp, 1, K, T)|\leq cT$$
for  some constant $c$ and all sufficiently small $T$.

\section{BS approximation revisited}
\label{BS-study}

In this section we  revisit  the BS approximation. In particular,
we obtain another  form of the approximation. The new approximation 
corrects a  certain drawback of the original one (to be explained).
Then we apply the modified BS approximation for estimation 
of implied volatility in the two-valued LVM.

Start with noting that
\begin{align}
\label{eqn:newApprox-C0}
\C(1,K,T)&\approx 2p\CBS(\sgp,1,K,T)
\text{ for } K>1\\
\label{eqn:newApprox-P0}
\Put(1,K,T)&\approx 2q\PBS(\sgm,1,K,T)
\text{ for } K<1.
\end{align}
In the  ATM case $K=1$ the equation for the call gives 
that $\Vatm(T)\approx 2p\BS(\sgp, T)$, while the equation for 
the put gives  that $\Vatm(T)\approx 2q\BS(\sgm, T)$. 
Below we obtain another form of the BS approximation in which this discrepancy is eliminated.

Using Theorem~\ref{Dupire_rep} 
we obtain that 
\begin{equation*}
\begin{split}
  \C(1,K,T)&=\sqrt{K}\int\limits_0^T\Vatm(T-s)h(s,\log K/\sgp)
e^{-\frac{1}{8} \sgp^2 s}ds\\
&=\sqrt{K} \int\limits_0^T\frac{\Vatm(T-s)}{\BS(\sgp, T-s)}
\BS(\sgp, T-s)h\big(s,\log K/\sgp\big)
e^{-\frac{1}{8} \sgp^2 s}ds,
\end{split}
\end{equation*}
where $\BS(\sgp, \cdot)$ is  the BS  ATM price (defined in~\eqref{BS-ATM}).
Approximating the time dependent ratio $\frac{\Vatm(T-s)}{\BS(\sgp, T-s)}$
by $\frac{\Vatm(T)}{\BS(\sgp,T)}$, i.e. 
by its value  at a single time moment $T$, gives 
  the new BS approximation for the price of the call option
\begin{equation}
\label{eqn:newApprox-C}
\begin{split}
 \C(1,K,T)&\approx\frac{\Vatm(T)\sqrt{K}}{\BS(\sgp,T)}
 \int\limits_0^T\BS(\sgp, T-s)h(s,\log K/\sgp)
e^{-\frac{1}{8} \sgp^2 s}ds\\
&=\frac{\Vatm(T)}{\BS(\sgp,T)}\CBS(\sgp,1,K,T)
\text{ for }  K>1.
\end{split}
\end{equation} 
Similarly,  we obtain the new approximation for the price of the put option, namely
\begin{align*}
\Put(1,K,T)&\approx \frac{\Vatm(T)}{\BS(\sgm,T)}\PBS(\sgm, 1,K,T)
\text{ for } K<1.
\end{align*}
In addition, note that the new BS approximation implies the approximation  given by~\eqref{eqn:newApprox-C0} and~\eqref{eqn:newApprox-P0}. For example, 
in the case of the call option we have that 
\begin{equation*}
\BS(\sgp, T)\approx \frac{\sgp\sqrt{T}}{\sqrt{2\pi}}\quad\text{and}\quad
\Vatm(T)\approx\frac{\sigatm\sqrt{T}}{\sqrt{2\pi}}=
2p\frac{\sgp\sqrt{T}}{\sqrt{2\pi}},
\end{equation*}
so that  $\frac{\Vatm(T)}{\BS(\sgp,T)}\approx 2p,$ which 
gives~\eqref{eqn:newApprox-C0}, as claimed.

Below we provide results of a numerical  experiment in which the new BS
approximation was used  to estimation  implied volatility.
In the experiment implied volatility (considered as a function 
of moneyness \(\frac{\log\left(K/S\right)}{\sigma (K,T)\sqrt{T}}\)) 
is estimated in a two-valued LVM with parameters $\sgp=0.2$ and 
$\sgm=0.9$.  
Figure~\ref{fig1} 
shows the  implied volatility smile for maturity $T=5$ years (Y)
obtained by numerical integration of the option pricing formulas in 
Theorem~\ref{Call-Put}
  and its approximation obtained by using equation~\eqref{eqn:newApprox-C}.

Figure~\ref{fig2} 
shows the difference between implied volatility  and 
its approximation for maturities $0.1Y$, $0.5Y$ and $5Y$.

\begin{figure}[htbp]
\centering
    \includegraphics[width=9cm, height=6.5cm]{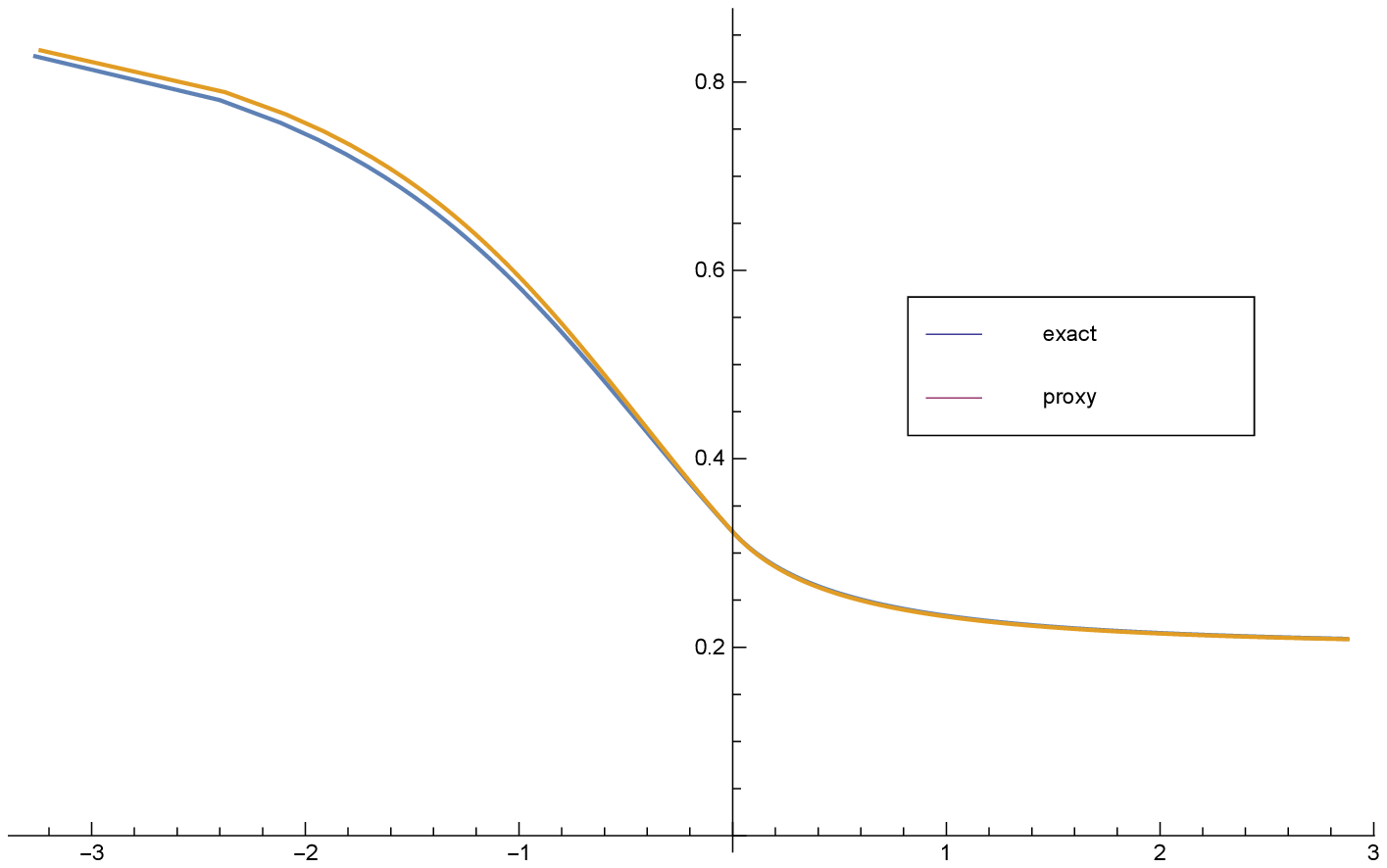}
\caption{{\footnotesize Implied volatility smile 
 and its approximation for maturity $5Y$}}
\label{fig1}
\end{figure}

\begin{figure}[htbp]
\centering
    \includegraphics[width=9cm, height=6cm]{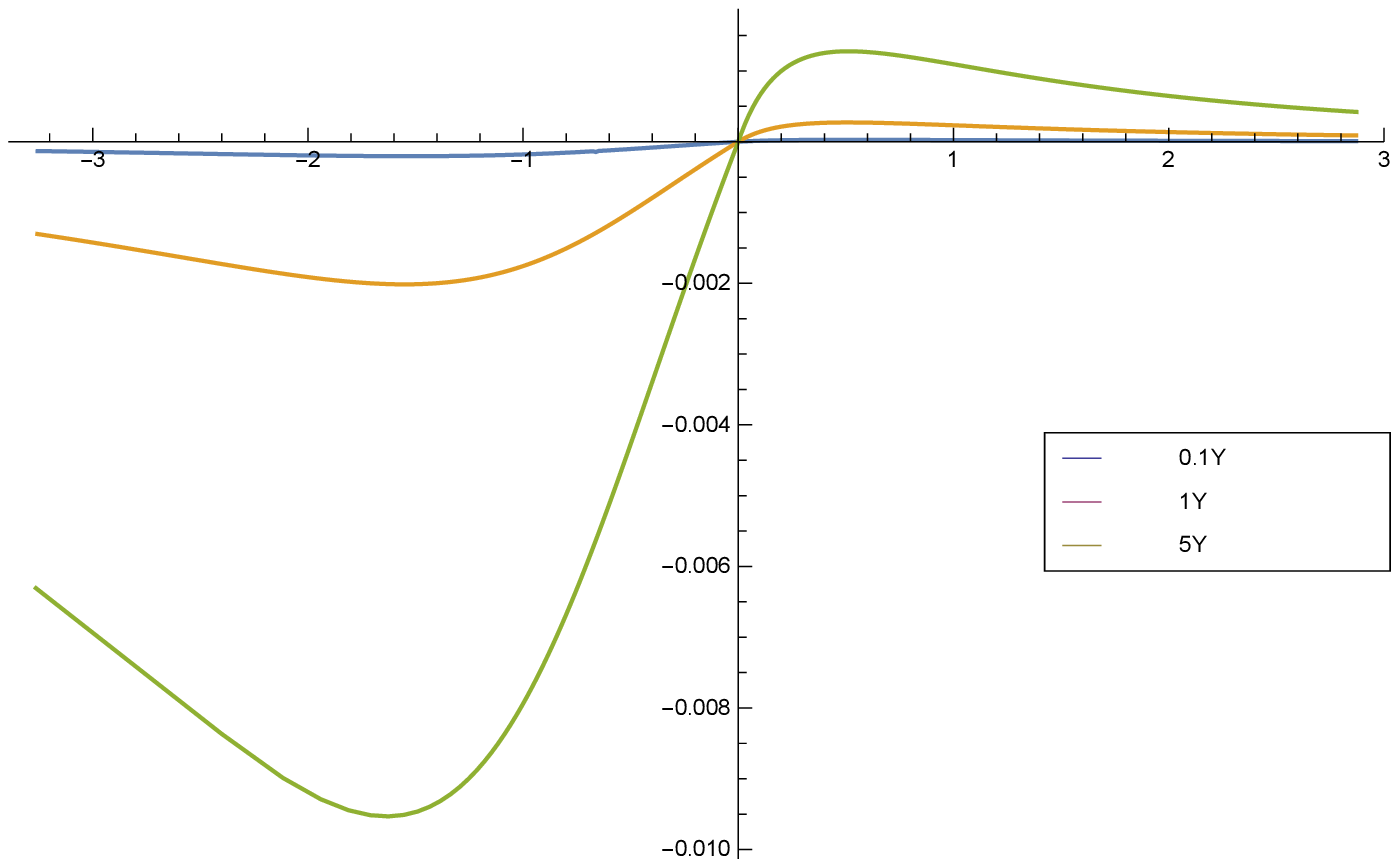}
\caption{{\footnotesize Plots of differences between implied volatility and 
its approximation for  maturities $0.1Y$, $0.5Y$ and $5Y$.
}}
\label{fig2}
\end{figure}

In addition,  it should be 
 noted that the new form of the BS approximation can be used 
  for the arbitrage free parametrization of implied volatility with given skew and constrained volatilities on wings (we do not discuss this here).

\section{Conclusion}

 In this  paper we consider  the  LVM  in which  volatility
takes two possible values. 
A particular value depends on the position of the underlying price with respect to a given threshold.
The model is well known, and a number of results have been obtained for the model in recent years. In particular, explicit pricing formulas for European options 
have been obtained  in Pigato~\cite{Pigato} in the case when the threshold
is taken at the money. These formulas have then been  used to establish 
 that the skew  explodes as  $T^{-1/2}$, as  maturity $T\to 0$,
 which   reproduces the  power law behaviour of the skew observed in some 
  real data.  The research method  in Pigato~\cite{Pigato} is based on the Laplace transform and related techniques.
  
In the present paper we propose another approach to the study 
of the two-valued LVM.  
Our approach  is based on the natural relationship of the two-valued 
LVM  with SBM and consists of using the joint
distribution of SBM and some of its functionals (\cite{GS}).  
We use our distributional results for 
obtaining   new  option pricing formulas and 
approximation of option prices  in terms of  the corresponding BS prices. 
The BS approximation is a key 
ingredient of our analysis  of implied volatility and the skew.
 Using this approximation allows to  obtain the aforementioned  behaviour
  of the implied volatility surface  by rather 
   elementary methods (e.g. Taylor's expansion of the second order).
 In addition, we show that the BS approximation can be improved
 and used to  estimate  implied volatility.

\section{Appendix. SBM with two-valued drift}
\label{B}

The process $X_t$ defined in~\eqref{X_t} is a special case of the  process 
$Z_t=(Z_t,\, t\geq 0)$ defined  as a strong solution of the  equation
\begin{equation}
\label{SBM_m}
dZ_t=m(Z_t)dt+(p-q)dL^{(Z)}_t+dW_t,
\end{equation}
where 
$m(z)=m_1{\bf 1}_{\{z\geq 0\}}+m_2{\bf 1}_{\{z<0\}}$,
$p\geq 0$ and $q\geq 0$ are given constants, such that $p+q=1$, 
$L_t^{(Z)}$ is  the local time of the process $Z_t$  at the origin (defined similarly to~\eqref{L}).
The existence and uniqueness of the strong solution of the equation is well known
(e.g. see Lejay~\cite{Lejay} and references therein).
In the special case $m_1=m_2=0$, the  process  $Z_t$ is 
SBM with parameter $p\in [0,1]$ (e.g., see Lejay~\cite{Lejay} and references therein).
By analogy, the process $Z_t$ can be called   SBM with the two-valued drift $m$.
 Theorem~\ref{C1} is  a special case of  the  theorem below.
\begin{theorem}[\cite{GS}, Theorem 2]
\label{T1}
Let $\tau^{(Z)}, V^{(Z)}$ and $L^{(Z)}_T$  be the last visit to the origin,
 the occupation time and the 
local time at the origin  of the process $Z_t$.
 If $Z_0=0$, then 
 the joint  density of random variables $\tau^{(Z)}, V^{(Z)}, Z_T$ and $L^{(Z)}_T$ is 
 given by
\begin{equation}
\label{f}
f_{T}(t, v, x, \ell)
=2\alpha(x)h(v, \ell p)h(t-v, \ell q)h(T-t, x)\beta(t,v,x,\ell)
\end{equation}
for $0\leq v\leq  t\leq T$ and $\ell\geq 0$, where the function $\alpha$ is defined in~\eqref{a(x)} and 
$$\beta(t,v,x,\ell)=
e^{-\frac{1}{2}(m_1^2v+m^2_2(t-v)+m^2(x)(T-t))-\ell(pm_1-qm_2)+m(x)x}.$$
\end{theorem}
\begin{example}
{\rm 
In the special case $m_1=m_2=m$ and $p=1/2$  equation~\eqref{f-1} becomes
\begin{equation*}
f_{T}(t, v, x, \ell)
=h(v, \ell/2)h(t-v, \ell/2)h(T-t, x)
e^{-\frac{m^2T}{2}+mx},
\end{equation*}
for $ 0\leq v\leq  t\leq T$ and $\ell\geq 0$, which is 
the joint density of quantities  $\tau^{(Z)}, V^{(Z)}, Z_T$ and $L^{(Z)}_T$ corresponding to the process 
$Z_t=mt+W_t$ in the case when  $Z_0=0$.
}
\end{example}

\begin{remark}
{\rm 
SDE~\eqref{X_SDE} is a special case of~\eqref{SBM_m}  with  the drift 
specified by  values 
$m_1=-\frac{\sgp}{2}$ and $m_2=-\frac{\sgm}{2}$
and probabilities  $p$ and  $q$ are  given by~\eqref{pq} and~\eqref{qq} respectively.
In this case we have that 
$pm_1-qm_2=
-\frac{\sgm\sgp}{2(\sgm+\sgp)}+
\frac{\sgp\sgm}{2(\sgm+\sgp)}=0,$ which 
reduces density~\eqref{f} to~\eqref{f-1}.
}
\end{remark}

\end{document}